\documentclass[10pt]{article}
\usepackage[top=1in, bottom=1in, left=1in, right=1in]{geometry}
\usepackage{amsmath, amssymb, amsthm}
\usepackage{color, xcolor, colortbl}
\usepackage{graphics, subfig, overpic}
\usepackage{hyperref}
\usepackage{algorithm, algpseudocode}
\usepackage[capitalize,nameinlink]{cleveref} %% use ``cref'' to replace ``ref
\crefname{equation}{}{Equations}
\Crefname{equation}{Eq.}{Equations}

\graphicspath{{images/}}
\numberwithin{equation}{section}
\newtheorem{theorem}{Theorem}
\newtheorem{lemma}[theorem]{Lemma}

\newtheorem{proposition}[theorem]{Proposition}

\newcommand\bbR{\mathbb{R}}

\newcommand\bbN{\mathbb{N}}

\newcommand\bbS{\mathbb{S}}

\newcommand\mB{\mathcal{B}}

\newcommand\mE{\mathcal{E}}
\newcommand\mG{\mathcal{G}}

\newcommand\mI{\mathcal{I}}
\newcommand\mJ{\mathcal{J}}

\newcommand\mL{\mathcal{L}}

\newcommand\mS{\mathcal{S}}

\newcommand\mU{\mathcal{U}}

\newcommand\tmu{\tilde{\mu}}

\newcommand{\tlam}{\tilde{\lambda}}

\newcommand\dd{\,\mathrm{d}} %differential operator
\newcommand\sffont[1]{{\sf{#1}}}

\newcommand\relu{{\sffont{ReLU}}}
\newcommand\id{{\sffont{id}}}
\newcommand\Convone{{\sffont{Conv1d }}}
\newcommand\Convtwo{{\sffont{Conv2d }}}

\newcommand\NN{\mathrm{NN}}
\newcommand\T{\mathsf{T}}
\newcommand\cnn{\mathrm{cnn}}

\newcommand\reverse[1]{{{\color{gray} }}}

\title{
  Solving Optical Tomography with Deep Learning
}

\date{}
\author{
  Yuwei Fan%
    \thanks{Department of Mathematics, Stanford University, Stanford, CA 94305. 
    Email: {\tt ywfan@stanford.edu}},~~
  Lexing Ying% 
    \thanks{Department of Mathematics and ICME, Stanford University, Stanford, CA 94305.
    Email: {\tt lexing@stanford.edu}}
}

\begin{document}
\maketitle

\begin{abstract}
  This paper presents a neural network approach for solving two-dimensional optical tomography (OT)
  problems based on the radiative transfer equation. The mathematical problem of OT is to recover
  the optical properties of an object based on the albedo operator that is accessible from boundary
  measurements. Both the forward map from the optical properties to the albedo operator and the
  inverse map are high-dimensional and nonlinear. For the circular tomography geometry, a
  perturbative analysis shows that the forward map can be approximated by a vectorized convolution
  operator in the angular direction. Motivated by this, we propose effective neural network
  architectures for the forward and inverse maps based on convolution layers, with weights learned
  from training datasets. Numerical results demonstrate the efficiency of the proposed neural
  networks.
\end{abstract}


%============================================================
\section{Introduction}\label{sec:intro}
Optical tomography (OT) is a non-invasive method for reconstructing the optical properties of the
medium from boundary measurements with harmless near-infrared light. A typical experiment is to
illuminate a highly-scattering medium by a narrow collimated beam and measure the light on the
surface by an array of detectors \cite{arridge2009optical}. Since it is non-destructive to
biological tissues, OT is of great interest in early tumor diagnosis in medicine, such as in brain
imaging \cite{boas2004diffuse} and breast imaging \cite{fantini2012near}. Other industrial
applications include atmospheric remote sensing \cite{weitkamp2006lidar} and semiconductor etching
\cite{edwards2012optically}, etc. We refer readers to the review paper \cite{arridge2009optical},
the book \cite{bal2019inverse} and references therein for more details of OT.

\paragraph{Background.}
The governing equation of the near-infrared light depends on the spatial scale, ranging from Maxwell
equations at the microscale, radiative transfer equation (RTE) at the mesoscale, and to diffusion
theory at the macroscale \cite{arridge2009optical}. Among them, RTE is the most widely accepted
model for light propagation in tissues. Let $\Omega\subset\bbR^n$ for $n=2$ or $3$ be a bounded
Lipschitz domain and $\bbS^{n-1}$ is the unit sphere in $\bbR^n$. Define $\Gamma_{\pm}=\{(x,v)\in
\partial \Omega\times \bbS^{n-1}\mid \pm v\cdot \nu(x)> 0\}$ with $\nu(x)$ to be the outward unit
normal to $\partial \Omega$ at $x$. The specific intensity $\Phi(x,v)$, defined as the intensity of
the light at the position $x$ in the direction $v$, satisfies the following RTE
\begin{equation} \label{eq:rte}
  \begin{aligned}
    v\cdot\nabla \Phi(x,v) + \mu_t(x)\Phi(x,v) = 
    \mu(x) \int_{\bbS^{d-1}}\sigma(v\cdot v')\Phi(x,v')\dd v' + Q(x, v),
    \quad & (x, v) \in \Omega \times \bbS^{n-1}, \\
    \Phi(x,v) = F(x,v), \quad & \text{ on }\Gamma_-.
  \end{aligned}
\end{equation}
The scattering phase function $\sigma$ satisfies $\int_{\bbS^{n-1}}\sigma(v\cdot v')\dd
v=1$. $Q(x,v)$ is the source inside $\Omega$ and $F(x,v)$ is the boundary condition specified at
$\Gamma_-$. In this paper, the internal light source is assumed to be absent, i.e., $Q(x, v) = 0$.
The transport coefficient $\mu_t(x)=\mu_a(x)+\mu(x)$ measures the total absorption, including the
physical absorption quantified by the term $\mu_a(x)$ and the scattering phenomenon quantified by
the term $\mu(x)$. Here we focus on the reconstruction of the scattering coefficient $\mu(x)$ under
the assumption that $\mu_a$ is a known constant.

The scattering phase function $\sigma(v\cdot v')$ describes the probability for a photon entering a
scattering process at the direction of propagation $v$ to leave this process at the direction
$v'$. The most common phase function in OT is the Henyey-Greenstein scattering function
\cite{HGscattering}
\begin{equation}\label{eq:HG}
  \sigma(v\cdot v') = \frac{1}{|\bbS^{n-1}|}
  \frac{1-g^2}{(1+g^2-2g v\cdot v')^{n/2}}.
\end{equation}
The parameter $g\in (-1, 1)$ defines the shape of the probability density. The case $g=0$ indicates
that the scattering is almost isotropic, whereas the value of $g$ close to $1$ indicates the
scattering is primarily a forward directed. A typical value in biological tissue is $g=0.9$.

The boundary condition in \eqref{eq:rte} guarantees the uniqueness of solutions of the RTE
\cite{case1963existence}.  In most applications, $F(x,v)$ is either a delta function (in $v$) at
direction $v=-\nu(x)$ or an angular-uniform illumination source. In both cases, $F(x,v)$ can be
written as an angular independent function $f(x) h(\nu(x)\cdot v)$ for some fixed distribution
$h(\cdot)$.

The measurement on the boundary can be angular dependent or independent. Here we focus on the
angular independent case, where the measurable quantity is given by
\begin{equation}\label{eq:measure}
  b(x) \equiv \mB\Phi(x) \equiv \int_{v\cdot\nu(x)>0} v\cdot\nu(x) \Phi(x, v)\dd v.
\end{equation}
The albedo operator is defined as
\begin{equation}\label{eq:albedo}
  \Lambda: H^k(\partial \Omega) \to H^{-k}(\partial \Omega), \quad
  f(x)\mid_{\partial \Omega} \to b(x)\mid_{\partial \Omega},
\end{equation}
where $k>2+n/2$. We refer the readers to \cite{zhao2018instability} for more details of the albedo
operator and the spaces $H^{\pm k}(\partial \Omega)$.

For a given $\mu(x)$, the albedo operator is a linear map, hence there exists a $\mu$-dependent
distribution kernel $\lambda (r, s)$ for $r, s\in \partial \Omega$ such that
\begin{equation}
  (\Lambda f )(r) = b(r) = \int_{\partial \Omega} \lambda (r, s) f (s) \dd S(s).
\end{equation}
The forward problem for the albedo operator is that, given the scattering coefficient $\mu$, to
compute the kernel $\lambda (r, s)$, i.e., $\mu\to\lambda $.  The inverse problem, which is central
to OT, is to recover the optical scattering coefficient $\mu$ in $\Omega$ based on the observation
data. Typically, the observation data is a collection of pairs $(f, \Lambda f)$ of the boundary
illumination source $f $ and the measurable quantity $\Lambda f$.  When the observation data is
sufficient, it is reasonable to assume that the kernel $\lambda $ is known and hence the inverse
problem is to recover $\mu$ from $\lambda $, i.e., $\lambda\to\mu$. The solvability of the inverse
problem has been well studied \cite{choulli1996reconstruction, stefanov2003optical,
  bal2008stability, bal2019inverse}. Since the measurements are angularly integrated, the inverse
problem is often sensitive to noise \cite{bi2015image, zhao2018instability}. For example, in the
diffusion limit where RTE can be approximated by a diffusion equation, the inverse problem is
considered ill-conditioned due to the elliptic nature \cite{arridge2009optical} of the equation.  In
other cases, the inverse problem can suffer H{\"o}lder instability due to its transport nature (see
\cite{zhao2018instability} for example).

From a computational perspective, both the forward and inverse problems associated with the albedo
operator \cref{eq:albedo} are numerically challenging. For the forward problem, since the unknown
field $\Phi(x,v)$ is a $(2n-1)$-dimensional function in both the space $x$ and the direction $v$,
direct solution of RTE is quite expensive even for the two-dimensional case. For OT problems, the
situation is worse since in each round of measurements the number of RTE solves is equal to the
number of light sources. For the inverse problem, the map $\lambda\to\mu$ is often numerically
unstable \cite{bi2015image,leng2015reconstruction} due to the ill-posedness and the measurement
noise. In order to avoid instability, an application-dependent regularization term is often required
in order to stabilize the inverse problem; see, for instance, \cite{hanke1997regularizing,
  chan1999nonlinear, kaipio1999inverse, gao2010multilevel, bi2015image}. Algorithmically, the
inverse problem is usually solved with iterative methods
\cite{hanke1997regularizing,haber2000optimization, gao2010multilevel, tong2018rte}, which often
require a significant number of iterations.

\paragraph{Contributions.}

In the recent years, deep neural networks (DNNs) have been very effective tools in a variety of
contexts and have achieved great successes in computer vision, image processing, speech recognition,
and many other artificial intelligence applications \cite{Hinton2012, Krizhevsky2012,
  goodfellow2016deep, MaSheridan2015, Leung2014, SutskeverNIPS2014, leCunn2015, SCHMIDHUBER2015}.
More recently, DNNs have been increasingly used in the context of scientific computing, particularly
in solving PDE-related problems
\cite{khoo2017solving,berg2017unified,han2018solving,fan2018mnn,Araya-Polo2018, Raissi2018,
  kutyniok2019theoretical, feliu2019meta}. First, since neural networks offer a powerful tool for
approximating high-dimensional functions \cite{cybenko1989approximation}, it is natural to use them
as an ansatz for high-dimensional PDEs \cite{rudd2015constrained, carleo2017solving, han2018solving,
  khoo2019committor, weinan2018deep}. A second main direction focuses on the low-dimensional
parameterized PDE problems, by using the DNNs to represent the nonlinear map from the
high-dimensional parameters of the PDE solution
\cite{long2018pde,han2017deep,khoo2017solving,fan2018mnn,fan2018mnnh2,fan2019bcr,li2019variational,bar2019unsupervised}.
Applying DNNs to inverse problems
\cite{khoo2018switchnet,hoole1993artificial,kabir2008neural,adler2017solving,lucas2018using,tan2018image,fan2019eit,raissi2019physics}
can be viewed as a particularly important case of this direction.

This paper applies the deep learning approach to the two-dimensional OT problems by representing
both the forward and inverse maps using neural network architectures.  The starting point of the new
architectures is reformulating RTE into an integral form, which allows for writing out explicitly
the forward map $\mu\to\lambda $. By applying a perturbative analysis on the forward map followed by
reparameterization, we find the forward map contains one-dimensional convolution in the angular
direction for the circular tomography geometry.  This observation motivates to represent the forward
map from 2D coefficient $\mu$ to 2D data $\lambda $ by a \emph{one-dimensional} convolution neural
network (with multiple channels). Following the idea of the back-projection method
\cite{feng2007levenberg}, the inverse map $\lambda\to\mu$ can be approximated by reversing the
architecture of the forward map followed with a simple two-dimensional neural network. For the test
problems being considered, the resulting neural networks have a relatively small number of
parameters, thanks to the convolutional structure. This rather small number of parameters allows for
rapid and accurate training, even on rather limited data sets, which is friendly for OT problems as
solving RTE is computationally quite expensive.

\paragraph{Organization.}
This rest of the paper is organized as follows. The mathematical background on the albedo operator
is studied in \cref{sec:math}. The design and architecture of the DNNs of the forward and inverse
maps are discussed in \cref{sec:nn}. Numerical tests are presented in \cref{sec:num}. 

%============================================================
\section{Mathematical analysis of the albedo operator}\label{sec:math}

The goal of this section is to make the relationship between the scattering field $\mu(x)$ and the
kernel $\lambda(r,s)$ of the albedo operator more explicit. The first step is to reformulate RTE as
an equivalent integral equation \cite{case1963existence, egger2014lp}. Denote by
\begin{equation}\label{eq:mJ}
  \mJ F(x, v) = \exp\left( -\int_0^t\mu_t(x-\tau v)\dd \tau \right) F(x-tv,v)
\end{equation}
the extension of boundary values, where $t(x, v)$ is the distance of a photon traveling from $x$ to
the domain boundary along the direction $-v$, i.e.,
\begin{equation} \label{eq:t}
  t(x, v) = \sup\{\tau: x-sv\in \Omega \text{ for } 0\leq s< \tau\}
\end{equation}
and $(x-t(x,v)v, v)\in\Gamma_{-}$. Introduce also the lifting operator
\begin{equation}\label{eq:mL}
  \mL Q(x, v) = \int_0^t\exp\left( -\int_0^\tau\mu_t(x-sv)\dd s \right) Q(x-\tau v, v)\dd \tau,
\end{equation}
and the scattering operator 
\begin{equation}\label{eq:mS}
  \mS \Phi(x,v) = \mu(x)\int_{\bbS^{n-1}}\sigma(v'\cdot v)\Phi(x, v')\dd v'.
\end{equation}
Direct calculations verify that
\begin{equation} \label{eq:operator_relation}
  \begin{aligned}
    (v\cdot\nabla+\mu_t) \mJ F &= 0,\\
    (v\cdot\nabla+\mu_t) \mL Q &= Q,\quad \mL Q\mid_{\Gamma_-}=0.
  \end{aligned}
\end{equation}
This indicates that the extension of the boundary value $\mJ F$ lies in the kernel of the transport
operator $v\cdot\nabla + \mu_t$ and the lifting operator is the right inverse of the transport
operator. Noticing that the internal source vanishes ($Q=0$), one can write RTE equivalently in an
integral form \cite{case1963existence}
\begin{equation}\label{eq:integral}
  \Phi = \mL \mS \Phi + \mJ F,
\end{equation}
which is a Fredholm integral equation of the second kind. The existence and uniqueness of the
integral equation is well understood \cite{case1963existence, LionsVol6} and inverting
\cref{eq:integral} results in
\begin{equation}\label{eq:solution}
  \Phi = (\mI - \mL\mS)^{-1}\mJ F,
\end{equation}
where $\mI$ is the identity operator.

In order to better understand the relationship between the scattering coefficient and the solution,
we perform a perturbative analysis for \cref{eq:solution}. Notice that all the operators $\mL$,
$\mS$ and $\mJ$ depend the scattering coefficient $\mu$ either directly or implicitly through
$\mu_t$. Denote the background of the scattering coefficients by $\mu_0$ and introduce the
perturbation
\[
\tmu \equiv \mu-\mu_0.
\]
Here we assume that both $\mu_0$ and $\mu_a$ are constant. The background of the total absorption
coefficient is then $\mu_{t,0}\equiv \mu_a+\mu_0$. In order to carry out the perturbative analysis,
we expand the operators $\mL, \mJ, \mS$ into terms of different orders of $\tmu$:
\begin{equation}
  \mL = \mL_0 + \mL_1 + \ldots, \quad
  \mJ = \mJ_0 + \mJ_1 + \ldots, \quad
  \mS = \mS_0 + \mS_1 + \ldots 
\end{equation}
where the {\em background operators} $\mL_0$, $\mS_0$ and $\mJ_0$ are independent of $\tmu$ while
$\mL_1$, $\mS_1$ and $\mJ_1$ are all linear in $\tmu$. With these new notations, \cref{eq:solution}
can be reformulated as
\begin{equation}
  \label{eq:Phiexp}
  \Phi = (\mI - \mL_0\mS_0-\mL_1 \mS_0- \mL_0\mS_1 - \ldots)^{-1}(\mJ_0+\mJ_1 + \ldots) F
\end{equation}
where $\ldots$ stands for higher order terms in $\tmu$.  Let us introduce
$\mE_1=\mL_1\mS_0+\mL_0\mS_1$, which is also first order in $\tmu$. When $\tmu$ is sufficiently
small, one can expand $(\mI - \mL_0\mS_0-\mL_1\mS_0-\mL_0\mS_1 - \ldots)^{-1}
=(\mI-\mL_0\mS_0-\mE_1-\ldots)^{-1} $ via a Neumann series
\begin{equation}
  (\mI - \mL_0\mS_0-\mE_1- \ldots)^{-1} = \mG_0 + \mG_0 \mE_1 \mG_0 + \ldots.  %\mG_0(\mE+\mEt)\mG_0(\mE+\mEt)\mG_0 + \ldots.
\end{equation}
where $\mG_0 = (\mI -\mL_0\mS_0)^{-1}$. Putting this back in \eqref{eq:Phiexp} and keeping only the
terms linear in $\tmu$, we conclude that the solution of RTE is approximated by
\begin{equation}\label{eq:approximate}
  \Phi \approx (\mG_0\mJ_0 + \mG_0\mJ_1 + \mG_0\mE_1\mG_0\mJ_0)F.
\end{equation}
Combining this with the measurement quantity \cref{eq:measure} results in 
\begin{equation}
  b = \mB\Phi = \mB(\mI - \mL\mS)^{-1}\mJ F
  \approx\mB (\mG_0\mJ_0 + \mG_0\mJ_1 + \mG_0\mE_1\mG_0\mJ_0)F.
\end{equation}
By introducing $b_0 = b\mid_{\mu=\mu_0}$, the boundary measurement obtained with the background
scattering coefficient $\mu_0$, it is equivalent to focus on the difference $b - b_0$. This is known
as \emph{difference imaging} in medical applications \cite{arridge2009optical} and the formula for
the difference is
\begin{equation}\label{eq:linearized}
  b - b_0 =
  \mB(\mI - \mL\mS)^{-1}\mJ F
  - \mB(\mI - \mL_0\mS_0)^{-1}\mJ_0 F
  \approx \mB \mG_0\mJ_1 F + \mB\mG_0\mE_1\mG_0\mJ_0 F.
\end{equation}

In practical applications, the boundary source can be represented as 
\begin{equation}
  F(x, v) = f(x) h(\nu(x)\cdot v).
\end{equation}
For example, if the boundary source is a laser, $h(\nu(x)\cdot v) = \delta(\nu(x)\cdot v-1)$; if the
source is angular independent, then $h(\nu(x)\cdot v) = 1 / |\bbS^{n-1}|$. Hence, the difference of
the albedo operator \cref{eq:albedo} applied to $f$ is
\begin{equation}
  b-b_0 = \left(\Lambda - \Lambda_{0}\right) f  = \mB(\mI - \mL\mS)^{-1}\mJ h f  
  - \mB(\mI - \mL_0\mS_0)^{-1}\mJ_0 h f  
  \approx 
  \mB \mG_0\mJ_1 h f  + \mB\mG_0\mE_1\mG_0\mJ_0h f .
\end{equation}
By setting $f(x)$ to be delta sources, one can extract from $b-b_0$ the kernel
\[
\tlam \equiv \lambda-\lambda_{0}
\]
of the difference albedo operator $\Lambda - \Lambda_{0}$.  In order to see $\tlam$ more explicitly,
denote the distribution kernel of the operator $\mG_0$ by $G_0(x, v, x', v')$, i.e., $\mG_0F(x, v) =
\int_{\Omega\times\bbS^{n-1}} G_0(x, v, x', v')F(x', v')\dd x'\dd v'$, and the distribution kernel
of the operator $\mE_1$ by $E_1(x, v, x', v')$.  By defining the operator $\beta[\eta](x, y) =
|x-y|\int_{0}^1\eta(x+\tau(y-x))\dd\tau$ for any function $\eta(x)$, the operator $\mB\mG_0\mJ$ can be 
represented as 
%first term on the right hand side of \eqref{eq:linearized} is
\begin{equation}
  \mB\mG_0\mJ F(x_r) = \int_{\nu(x_r)\cdot v>0}\hspace{-30pt}\nu(x_r)\cdot v \int_{\Omega}
  \int_{\partial \Omega} G_0\left(x_r, v, x, \widehat{x-x_s}\right) \exp(-\beta[\mu_t](x_s, x) )
  F(x_s, \widehat{x-x_s} )\dd v\dd x\dd x_s,
\end{equation}
where $\hat{x} = \frac{x}{|x|}$. Using the approximation
\begin{equation}
  \begin{aligned}
    \exp(-\beta[\mu_t](x, y)) &= \exp\left(-\beta[\mu_{t,0}](x,y)\right)
    \exp\left(-\beta[\tmu](x,y)\right)\\
                              & \approx \exp\left(-\beta[\mu_{t,0}](x,y)\right)
                              \left(1 - \beta[\tmu](x,y)\right),
  \end{aligned}
\end{equation}
from $\tmu = \mu - \mu_0 = \mu_t - \mu_{t,0}$, the kernel of the first term $\mB\mG_0\mJ_1$ is 
%can be approximated up to the first-order in $\tmu$ as
\begin{equation}\label{eq:d1}
  d_1(x_r, x_s) = 
  -\int_{\nu(x_r)\cdot v>0}\hspace{-30pt}\nu(x_r)\cdot v
  \int_{\Omega} G_0\left(x_r, v, x, \widehat{x-x_s}\right) 
  \exp\left(-\beta[\mu_{t,0}](x_s, x) \right)
  \beta[\tmu](x_s, x)
  h(\nu(x_s)\cdot \widehat{x-x_s} )\dd v\dd x,
\end{equation}
which is linear in $\tmu$ through $\beta[\tmu]$. Similarly, the kernel of the second term
$\mB\mG_0\mE_1\mG_0\mJ_0$ can be approximated by
\begin{equation}  \label{eq:d2}
  \begin{aligned}
    d_2(x_r, x_s) &= 
    \int_{\nu(x_r)\cdot v>0}\hspace{-30pt}\nu(x_r)\cdot v
    \int_{\Omega^3}\int_{(\bbS^{n-1})^2}
    G_0(x_r,v,x_1,v_1)E_1(x_1,v_1,x_2,v_2)G_0(x_2,v_2,x_3,\widehat{x_3-x_s})\\
    &\qquad\times
    \exp\left(-\beta[\mu_{t,0}](x_3,x_s)\right)h(\nu(x_s)\cdot\widehat{x_3-x_s})
    \dd x_1\dd x_2\dd x_3\dd v\dd v_1\dd v_2,
  \end{aligned}
\end{equation}
which is also linear in $\tmu$ through $E_1(x_1,v_1,x_2,v_2)$. Putting them together, the kernel of
the difference of the albedo operator $\Lambda -\Lambda_{0}$ can then approximated by
\begin{equation}\label{eq:def_d}
  \tlam(x_r,x_s) \equiv (\lambda-\lambda_{0})(x_r,x_s) \approx d(x_r, x_s) \equiv d_1(x_r, x_s) + d_2(x_r, x_s).
\end{equation}

%\Cref{eq:integral} can be also reformulated as 
%\begin{equation}
%  \Phi = \mJ f + \mL\mS\mJ f + (\mI-\mL\mS)^{-1}\mL\mS\mL\mS f,
%\end{equation}
%where $\Phi_0:=\mJ f$ is the ballistic component, $\Phi_1:=\mL\mS\mJ f$ denotes by the single
%scattering component and $\Phi_2 := \Phi - \Phi_0 - \Phi_1$ is the multiple scattering component.
%If the problem is subcritical such that the first two components fails to depict the photon
%transport, one may express \cref{eq:integral} in term of the Neumann series as
%\begin{equation}
%  u = \sum_{k=0}^{\infty}(\mL\mS)^k\mJ f.
%\end{equation}
%The term $k=0$ is the ballistic component, and the term $k$ is the $k$ times scattering component.

%============================================================
\section{Neural networks for OT}\label{sec:nn}

The discussion below focuses on the two-dimensional case, i.e., $n=2$. For circular tomography
geometry, the domain $\Omega$ is a unit disk
\cite{arridge2009optical,bi2015image,tong2018rte,bal2019inverse}. As illustrated in
\cref{fig:domain}, the light sources are placed on the boundary equidistantly, while the receivers
are shifted by a half spacing. The forward problem of OT is to determinate all the outgoing
intensity on the receivers when the light source is activated one by one. The measured data is the
kernel $\lambda(x_r,x_s)$, where $x_s=(\cos(s),\sin(s))$ with $s=\frac{2\pi k}{N_s}$,
$k=0,\dots,N_s-1$ and $x_r=(\cos(r),\sin(r))$ with $r=\frac{(2j+1)\pi}{N_r}$, $j=0,\dots, N_r-1$,
where $N_s=N_r$ in the current setup. Both the absorption coefficient $\mu_a$ and the background
scattering coefficient $\mu_0$ are assumed to be known constants. The inverse problem of OT is to
recover the scattering coefficient $\mu$ in the domain given the observation data $\lambda(x_r, x_s)
- \lambda_0(x_r, x_s)$, where $\lambda_0(x_r, x_s)$ is the measurement data of the medium with
scattering coefficient to be $\mu_0$.

\begin{figure}[htb]
  \centering
  \includegraphics[width=0.40\textwidth]{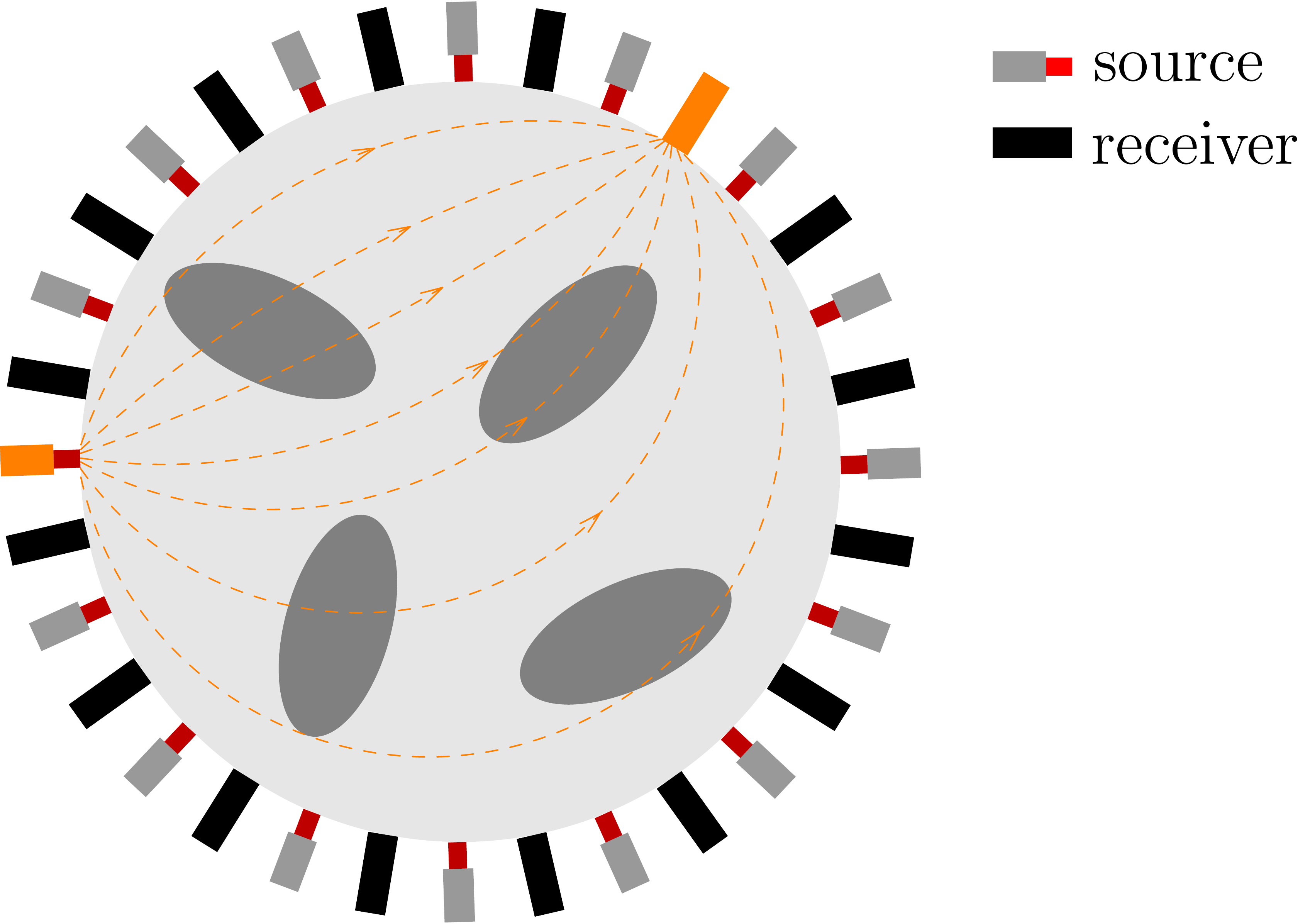}
  \label{fig:domain}
  \caption{Illustration of the problem setup. The domain is a unit disk and the
    light sources and the receivers are equidistantly placed on the boundary, with a half spacing
    shift in between.}
\end{figure}

\subsection{Forward problem of OT}
Since the domain $\Omega$ is a disk, it is convenient to write the problem in the polar coordinates.
Let $x_r=(\cos(r), \sin(r))$, $x_s=(\cos(s), \sin(s))$ and $x=(\rho\cos(\theta), \rho\sin(\theta))$,
where $\rho\in[0,1]$ denotes the radial direction and $r,s,\theta\in[0, 2\pi)$ denotes the
  angular direction.
% We also represent the direction $v$ as $v=(\cos(\psi), \sin(\psi))$.

\begin{figure}[htb]
    \centering
    \subfloat[$\lambda(x_r, x_s)$]{
      \includegraphics[height=0.13\textheight]{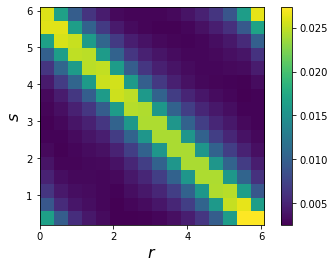}
    }
    \subfloat[$(\lambda-\lambda_0)(x_r, x_s)$]{
      \includegraphics[height=0.13\textheight]{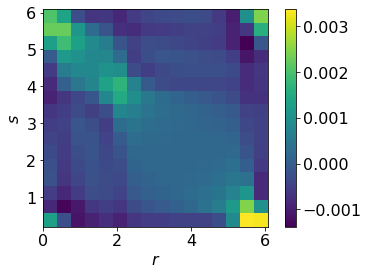}
    }
    \subfloat[$\lambda(x_h, x_s)$]{
      \includegraphics[height=0.13\textheight]{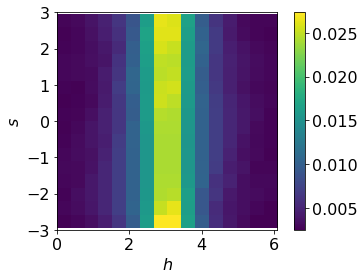}
    }
    \subfloat[$(\lambda-\lambda_0)(x_h, x_s)$]{
      \includegraphics[height=0.13\textheight]{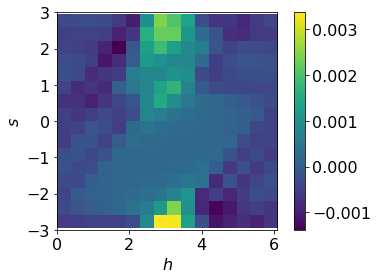}
    }
    \label{fig:measurement} 
    \caption{The uppers figures are the measurement data $\lambda(x_r, x_s)$
      and the difference $\lambda(x_r, x_s)-\lambda_0(x_r, x_s)\approx d(x_r, x_s)$ with respect to
      the background, respectively. The horizontal and vertical axes are $s$ and $r$,
      respectively. The lower figures are shift of their upper figures by $h=r-s$.}
\end{figure}

\paragraph{Convolution in the angular direction.}
\Cref{fig:measurement} presents an example of the measurement data $\lambda(x_r, x_s)$ and
$\lambda(x_r, x_s)-\lambda_0(x_r,x_s)$. Notice that the main signal concentrates upon the diagonal
part and the left-lower and right-upper corners. Due to the circular tomography geometry, it is
convenient to ``shear'' the measurement data by introducing a new angular variable $h = r - s$,
where the difference here is understood modulus $2\pi$. As we shall see, this shearing step
significantly simplifies the architecture of the NNs. Under the new parameterization, the measurement
data is
\begin{equation}
  d(h, s) \equiv d(x_{h+s}, x_s) = d( (\cos(s+h),\sin(s+h)), (\cos(s),\sin(s))).
\end{equation}
By also writing $\mu(\rho, \theta) \equiv \mu( (\rho\cos(\theta), \rho\sin(\theta)))$ in the polar
coordinates, the linear dependence of $d(h,s)$ on $\tmu$ in \eqref{eq:d1} and \eqref{eq:d2} states
that there exists a kernel distribution $K(h, s, \rho, \theta)$ such that
\begin{equation}
  d(h, s) = \int_{0}^1\int_{0}^{2\pi}K(h, s, \rho, \theta) \tmu(\rho,
  \theta)\dd\rho\dd\theta.
\end{equation}
The following proposition states that this can in fact be written as a convolution in the angular
direction.
%Next we study the properties on the kernel $K(h, s, \rho, \theta)$, which is the basis of the
%neural network architecture.

\begin{proposition}\label{pro:convolution}
  There exists a function $\kappa(h, \rho, \cdot)$ periodic in the last argument such that 
  \begin{equation}\label{eq:convolution}
    K(h, s, \rho, \theta) = \kappa(h, \rho, s-\theta).
  \end{equation}
\end{proposition}
The proof of this proposition uses some basic formulas summarized in the following lemma.
\begin{lemma}
  If $R\in\bbR^{2\times 2}$ is a rotation matrix, then 
  \begin{align}
      \label{eq:rotation_beta}
        \beta[\mu_0](Rx, Ry) &= \beta[\mu_0](x, y),\\
      \label{eq:rotation_G}
        G_0(x, v, x', v') &= G_0(Rx, Rv, Rx', Rv').
  \end{align}
\end{lemma}
\begin{proof}
  The definition of $\beta$ indicates $\beta[\mu_0](x,y) = \mu_0|x-y|$. Thus \cref{eq:rotation_beta}
  holds.

  Denote the distribution kernel of $\mL_0\mS_0$ by $L_0(x, v, x', v')$.
  Since $\mG_0 = (\mI - \mL_0\mS_0)^{-1}$, we just need to check $L_0(Rx, Rv, Rx', Rv') 
  = L_0(x, v, x', v')$. 
  Notice that the operator $\mL\mS$ is defined as 
  \begin{equation}\label{eq:LS}
    \begin{aligned}
      \mL\mS Q(x, v) &= \int_0^t\exp\left( -\beta[\mu_t](x, x-\tau v) \right)\mu(x-\tau v)
      \int_{\bbS^{1}}\sigma(v\cdot v')Q(x-\tau v, v')\dd v' \dd \tau\\
      &= \int_{\Omega}\int_{\bbS^1}\frac{\delta(v-\widehat{x-x'})}{|x-x'|}
      \exp\left( -\beta[\mu_t](x,x') \right)\mu(x')\sigma(v\cdot v')Q(x', v')\dd x'\dd v'.
    \end{aligned}
  \end{equation}
  Then the kernel $L_0$ reads 
  \begin{equation*}
    L_0(x, v, x', v') = 
    \frac{\delta(v-\widehat{x-x'})}{|x-x'|}
    \exp\left( -\beta[\mu_{t,0}](x,x') \right)\mu_0\sigma(v\cdot v').
  \end{equation*}
  Since $\delta(Rv - R(\widehat{x-x'})) = \delta(v-\widehat{x-x'})$, 
  $Rv \cdot Rv'=v\cdot v'$ and \cref{eq:rotation_beta}, one can directly obtain
  \begin{equation*}
    L_0(Rx, Rv, Rx', Rv') =  L_0(x, v, x', v').
  \end{equation*}
  This completes the proof.
\end{proof}

\begin{proof}[Proof of \cref{pro:convolution}]
  To prove \cref{eq:convolution}, one needs to show that, for any $\rho\in[0,1)$ and any
  $h\in[-\pi,\pi)$ and $s, \psi\in[0, 2\pi)$,
  \begin{equation}\label{eq:d}
    d(h, s+\psi) = \int_{0}^1\int_{0}^{2\pi}K(h, s, \rho, \theta) 
    \tmu(\rho,\theta+\psi)\dd\rho\dd\theta
  \end{equation}
  holds. Notice \cref{eq:def_d} that $d$ has two parts. We study them one by one.
  Define the rotation matrix
  $R = \begin{pmatrix} \cos(\psi) & -\sin(\psi)\\ \sin(\psi) & \cos(\psi) \end{pmatrix}$, then
  \begin{equation*}
    \begin{aligned}
      d_1(h, s+\psi) = d_1(Rx_r, Rx_s) &=
      -\int_{\nu(R x_r)\cdot v>0}\hspace{-30pt}\nu(R x_r)\cdot v
      \int_{\Omega} G_0\left(R x_r, v, x, \widehat{x-Rx_s}\right) \\
      &\qquad \times
      \exp\left(-\beta[\mu_{t,0}](Rx_s, x) \right)
      \beta[\tmu](Rx_s, x)
      h(\nu(R x_s)\cdot \widehat{x-Rx_s} )\dd v\dd x.
    \end{aligned}
  \end{equation*}
  Since $\Omega$ is a disk, the integral keeps unchanged if we change of variables as $v\to R v$ and
  $x\to R x$.  Using $\nu(Rx_r)\cdot Rv=\nu(x_r)\cdot v$, \cref{eq:rotation_G} and
  \cref{eq:rotation_beta} to eliminate the rotation and changing the variable again as $R v\to v$
  and $R x\to x$, we obtain 
  \begin{equation*}
    \begin{aligned}
      d_1(h, s+\psi) = d_1(Rx_r, Rx_s) &=
      -\int_{\nu(x_r)\cdot v>0}\hspace{-30pt}\nu(x_r)\cdot v
      \int_{\Omega} G_0\left(x_r, v, x, \widehat{x-x_s}\right) \\
      &\qquad \times
      \exp\left(-\beta[\mu_{t,0}](x_s, x) \right)
      \beta[\tmu](Rx_s, Rx) h(\nu(x_s)\cdot \widehat{x-x_s} )\dd v\dd x.
    \end{aligned}
  \end{equation*}
  This completes the proof of the $d_1$ part.

  Next we study the second part $d_2$. Noticing \cref{eq:LS}, we obtain the kernel distribution
  $E_1(x, v, x', v')$ 
  \begin{equation*}
      E(x, v, x', v') = 
      \frac{\delta(v-\widehat{x-x'})}{|x-x'|}
      \exp\left( -\beta[\mu_{t,0}](x,x') \right)\sigma(v\cdot v')
      \left(-\beta[\tmu](x, x') + \tmu(x') \right).
  \end{equation*}
  Using \cref{eq:rotation_beta}, we have
  \begin{equation*}
      E_1(Rx, Rv, Rx', Rv') = 
      \frac{\delta(v-\widehat{x-x'})}{|x-x'|}
      \exp\left( -\beta[\mu_{t,0}](x,x') \right)\sigma(v\cdot v')
      \left(-\beta[\tmu](R x, R x') + \tmu(R x') \right).
  \end{equation*}
  Then using the same technique in the proof of the first part, we can show that \cref{eq:d} also 
  holds for the second part. This completes the proof.
\end{proof}

\cref{pro:convolution} shows that $K$ acts on $\tmu$ in the angular direction by a convolution,
i.e.,
\begin{equation}\label{eq:convolution_d}
  d(h, s) = \int_0^1(\kappa(h, \rho,\cdot) * \tmu(\rho, \cdot))(s)\dd\rho.
\end{equation}
This effectively reduces the forward map to a family of 1D convolutions, parameterized by $\rho$ and
$h$.

Till now all the analysis is in the continuous space. One can apply a discretization on the RTE
\cref{eq:rte} by the finite volume method on the space and discrete velocity method on the direction
domain \cite{gao2009fast}. The kernel distribution $G_0$ and $E$ are replaced by its discrete
version. The actual discretization is often problem-dependent and we leave it to \cref{sec:num}.
Here with a slight abuse of notation, we use the same letters to denote the continuous kernels,
variables and their discretization. Then the discretization version of \cref{eq:convolution_d} is
\begin{equation}\label{eq:discrete_d}
  d(h, s) \approx \sum_{\rho}(\kappa(h,\rho, \cdot) * \tmu(\rho, \cdot))(s).
\end{equation}

%For simplification, we denote $N_{\phi}$ by the number of the discretization point on the variable
%$\phi$ and assume $N_s = N_{\theta}$.

\paragraph{Neural network architecture.}
The perturbative analysis shows that if $\tmu$ is sufficiently small, the forward map
$\tmu(\rho,\theta)\to \tlam(h, s)$ can be approximated by \cref{eq:discrete_d}. This indicates that
the forward map \cref{eq:discrete_d} can be approximated by a convolution layer for small
$\tmu$. For larger $\tmu$, this linear approximation is no longer accurate. In order to extend the
neural network for \cref{eq:discrete_d} to the nonlinear case, we propose to increase the number of
convolution layers and include nonlinear activation functions, as shown in \cref{alg:forward}.  Here
$\Convone[\alpha, w, \relu]$ stands for a one-dimensional layer with channel number $\alpha$, window
size $w$, and activation function as $\relu$. Note that because the value of the measurement data
ranges in $\bbR$, no activation function is applied after the last layer. Since the convolution in
\cref{eq:discrete_d} is global, the architectural parameters are chosen with
\begin{equation}
  w N_{\cnn} \geq N_s
\end{equation}
so that the resulting network is capable of capturing global interactions. When $N_s$ is large, it
is possible that the recently proposed multiscale neural networks, for example MNN-$\mathcal{H}$-net
\cite{fan2018mnn}, MNN-$\mathcal{H}^2$-net \cite{fan2018mnnh2}, and BCR-net \cite{fan2019bcr}, are
more efficient for such global interactions. However in order to simplify the presentation, the
discussion here sticks to the convolutional layers.

\begin{algorithm}[htb]
  \begin{small}
    \begin{center}
      \begin{algorithmic}[1]
        \Require $\alpha$, $w$, $N_{\cnn}\in\bbN^+$, $\tmu\in\bbR^{N_{\rho}\times N_{\theta}}$
        \Ensure $\tlam \in\bbR^{N_h\times N_s}$
        \State $\xi^{(0)} = \tmu$ with $\rho$ as the channel direction
        \For {$k$ from $1$ to $N_{\cnn}-1$ by $1$}
        \State $\xi^{(k)} \leftarrow \Convone[\alpha, w, \relu](\xi^{(k-1)})$
        \EndFor
        \State $\tlam \leftarrow \Convone[N_{h}, w, \id](\xi^{(N_{\cnn}-1)})$
        \State \sffont{return} $\tlam$
      \end{algorithmic}
    \end{center}
  \end{small}
  \caption{\label{alg:forward} Neural network architecture for the forward problem
  $\tmu\to \tlam$.}
\end{algorithm}

\subsection{Inverse problem of OT}
The perturbative analysis shows that if $\tmu$ is sufficiently small, the forward map can
be approximated by 
\begin{equation}
    \tlam \approx  K \tmu,
\end{equation}
which is the operator notation of the discretization \cref{eq:discrete_d}. Here $\tmu$ is a vector
indexed by $(\rho, \theta)$, $\tlam$ is a vector indexed by $(h, s)$, and $K$ is a matrix with row
indexed by $(h, s)$ and column indexed by $(\rho, \theta)$. The filtered back-projection method
\cite{feng2007levenberg} suggests the following formula to recover $\tmu$:
\begin{equation}
  \tmu\approx (K^\T K + \epsilon I)^{-1} K^\T \tlam.
\end{equation}

Since $K^\T\tlam$ can also be written as a family of convolutions
\begin{equation}
  (K^\T \tlam)(\rho,\theta) = \sum_{h} (\kappa(h,\rho, \cdot) * \tlam(h,\cdot))(\theta),
\end{equation}
the application of $K^\T$ to $\tlam$ can be approximated with a one-dimensional convolutional neural
network, similar to $K$. For the part $K^\T K + \epsilon I$, which can be viewed as a
post-processing in the $(\rho,\theta)$ space, we implement this with several two-dimensional
convolutional layers for simplicity. However, for problems with larger sizes, multiscale neural
networks such as \cite{fan2018mnn,fan2018mnnh2,fan2019bcr} can be also used. The resulting
architecture for the inverse map is summarized in \cref{alg:inverse} and illustrated in
\cref{fig:inverse}

\begin{algorithm}[htb]
  \begin{small}
    \begin{center}
      \begin{algorithmic}[1]
        \Require $\alpha_1, \alpha_2$, $w_1, w_2$, $N_{\cnn_1}$, $N_{\cnn_2}\in\bbN^+$, 
        $\tlam\in\bbR^{N_h\times N_s}$
        \Ensure $\tmu\in\bbR^{N_{\rho}\times N_{\theta}}$
        \State $\zeta^{(0)} = \tlam$ with $h$ as the channel direction
        \For {$k$ from $1$ to $N_{\cnn_1}$ by $1$}
        \State $\zeta^{(k)} \leftarrow \Convone[\alpha_1, w_1, \relu](\zeta^{(k-1)})$
        \EndFor
        %\State $\xi^{(0)} \leftarrow \Interpolate(d^{(N_{\cnn_1})})$ 
        \State $\xi^{(0)} \leftarrow \zeta^{(N_{\cnn_1})}$
        \For {$k$ from $1$ to $N_{\cnn_2}-1$ by $1$}
        \State $\xi^{(k)} \leftarrow \Convtwo[\alpha_2, w_2, \relu](\xi^{(k-1)})$
        \EndFor
        \State $\tmu \leftarrow \Convtwo[1, w_2, \id](\xi^{(N_{\cnn_2}-1)})$
        \State \sffont{return} $\tmu$
      \end{algorithmic}
    \end{center}
  \end{small}
  \caption{\label{alg:inverse} Neural network architecture for the inverse problem
  $\tlam \to \tmu$.}
\end{algorithm}

\begin{figure}[htb]
  \centering
  \includegraphics[width=\textwidth]{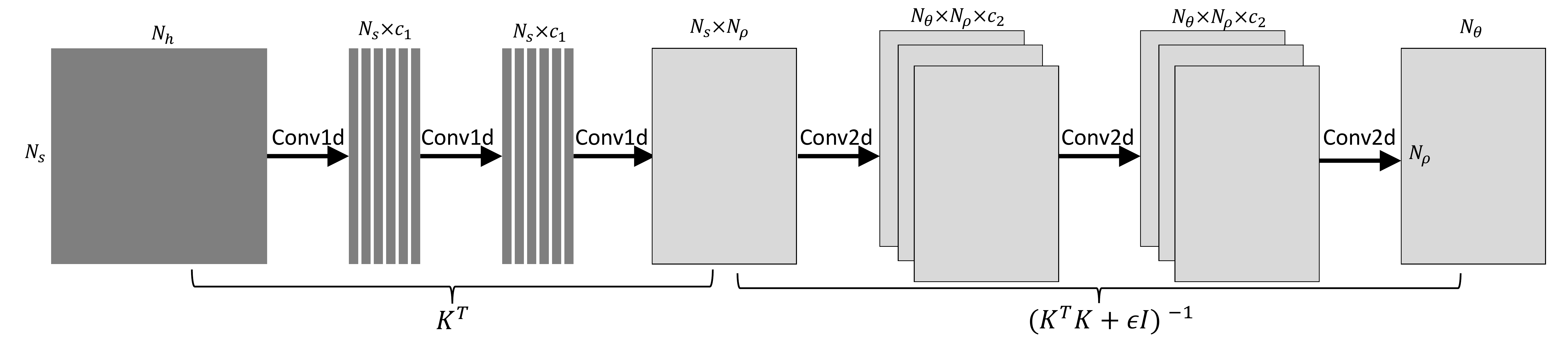}
  \label{fig:inverse} 
  \caption{Neural network architecture for the inverse map of OT.}
\end{figure}

\section{Numerical tests} \label{sec:num}

This section reports the numerical performance of the proposed neural network architectures for the
forward and inverse maps.

%We start by detailing the setup of the numerical discretization of RTE and the setup of the neural
%network.

\subsection{Experimental setup}
The RTE in \cref{eq:rte} is discretized with a finite volume method in $x$ and a discrete velocity
method in $v$. The upwind scheme is used for the convection term and the composite trapezoidal rule
is applied for the integral of the scattering term. The value of $\sigma(v'\cdot v)$ is replaced by
its value on the discretization points with a scaling such that its numerical quadrature is $1$. The
multi-level method proposed in \cite{gao2009fast} is adopted to solve the discrete system. The
domain $\Omega$ is partitioned by triangle mesh with $6976$ elements and $3553$ points. The
direction $v$ is uniformly discretized using $32$ points. In the polar coordinates, the domain
$(\rho, \theta)\in [0,1]\times [0,2\pi)$ is partitioned by a uniformly Cartesian mesh with $96\times
192$ points. As a technical note, since \cref{alg:forward,alg:inverse} are designed for the
scattering coefficient in the polar coordinates, the scattering coefficient on the triangle mesh
is treated as a piece-constant function and it is further interpolated on to the polar grid.

To mimic the setup of realistic medical applications, $\Omega$ is a disc with the radius equal to
$20\mathrm{mm}$ and the background scattering and absorption coefficient are $1 \mathrm{mm}^{-1}$
and $0.01 \mathrm{mm}^{-1}$, respectively \cite{gao2010multilevel, bi2015image, tong2018rte}. 
The parameter $g$ in \cref{eq:HG} is set as $g=0.9$, a typical value of biological tissues.
In the experiment, $N_s=16$ light sources and $N_r=16$ receivers are equidistantly placed on the
boundary of the domain with a half spacing shift (see \cref{fig:domain}). The source light is an
angular independent pointolite, i.e., the $s$-th light source is $F(x, v) = \delta(x-x_s)$.

The NN is implemented with Keras \cite{keras} running on top of TensorFlow \cite{tensorflow}. Nadam
is chosen as the optimizer \cite{dozat2015incorporating} and the mean squared error is used as the
loss function. The parameters of the network are initialized by Xavier initialization
\cite{glorot2010understanding}.
In the training process, the batch size and the learning rate is firstly set as $16$ and $10^{-3}$
respectively, and the NN is trained $100$ epochs. Then we increase the batch size by a factor $2$
till to $256$ with the learning rate unchanged, and then decrease the learning rate by a factor
$10^{1/2}$ to $10^{-5}$ with the batch size fixed as $256$.  In each step, we train the NN $50$
epochs. The selection of the channel number $\alpha$, number of convolution layers $N_{\cnn}$ and
the window size $w$ will be discussed in the numerical results.

\subsection{Numerical results}
%\LY{maybe we should use $\lambda$ instead of $d$, since $d$ is the linear approximation and it is
%  not what we compute or extract}

For a fixed scattering coefficient field $\mu$, $\lambda (h, s) =\lambda ((\cos(s+h),\sin(s+h)),
(\cos(s),\sin(s)))$ stands for the \emph{exact} measurement data solved by numerical discretization
of \cref{eq:rte}. The prediction of the forward NN from $\mu$ is denoted by $\lambda^{\NN}$, while
the one of the inverse NN from $\lambda$ is denoted by $\mu^{\NN}$. The accuracy for the forward
problem is measured by the relative error in the $\ell^2$ norm:
\begin{equation}\label{eq:relativeerror}
  %\frac{\|\mu-\mu^{\NN}\|_{\ell^2}}{\|\mu\|_{\ell^2}},\quad
  \frac{\|\lambda -\lambda ^{\NN}\|_{\ell^2}}{\|\lambda \|_{\ell^2}}.
\end{equation}
For each experiment, the test error is then obtained by averaging \cref{eq:relativeerror} over a
given set of test samples. The numerical results presented below are obtained by repeating the
training process three times, using different random seeds for the NN initialization.

The scattering coefficient $\mu(x)$ is assumed to be piecewise constant. For each sample $\mu(x)$,
we randomly generate $N_{e}$ ellipses in $\Omega$ and set $\mu(x)=2 \mathrm{mm}^{-1}$ in the
ellipses and $1 \mathrm{mm}^{-1}$ otherwise. For each ellipse, the width and height are sampled from
the uniform distributions $\mU(0.0075, 0.015)$ and $\mU(0.00375, 0.0075)$, respectively, the
direction is uniformly random over the unit circle, and the position is uniformly sampled in the
disk. It is also required that each ellipse lies in the disk and there is no intersection between
each two ellipses. For each test, $10,204$ samples $\{(\mu_i,\lambda_i)\}$ are generated with $8192$
used for training and the remaining 2048 for testing.

While \cref{alg:forward,alg:inverse} assume for simplicity that $N_{\theta}=N_s$, this is often not
the case in the experimental setup. To deal with this issue, for the forward problem we first
compress $\mu$ from $N_{r}\times N_{\theta}$ to $\alpha\times N_{s}$ by a one-dimensional
convolution layer with channel number $\alpha$, window size $N_{r}/N_s$, and strides $N_{r}/N_s$.
For the inverse problem, an interpolation operator for extending the data of size $\alpha\times N_s$
to $N_r\times N_{\theta}$ is added after the one-dimensional convolution neural networks.  In the
implementation, the interpolation is implemented by two layers. The first layer interpolates the
data of size $\alpha\times N_{s}$ along with the angular direction to $\alpha\times N_{\theta}$ by a
one-dimensional convolution layer with channel number $\alpha\times N_{\theta} / N_s$ and window
size $1$, and a column major reshape.  The second layer interpolates the data of size $\alpha\times
N_{\theta}$ along with the radial direction to $N_r\times N_{\theta}$ by a convolution layer with
channel number $N_r$ and window size $1$.

\paragraph{Forward problem.}
The data set is generated with the number of ellipses $N_e=4$ and the window size $w$ in
\cref{alg:forward} is set to be $5$. Multiple numerical experiments are performed to study how the
test error depends on the channel number $\alpha$ and the convolution layer number $N_{\cnn}$, with
the results presented in \cref{fig:fwderr}. As the number of channels increases, the test error
first consistently decreases and then saturates. The same is observed for the number of convolution
layers. The choices of the hyper-parameters $\alpha=32$ and $N_{\cnn}=8$ offers a reasonable balance
between accuracy and efficiency. For this specific case, the number of parameters is $7.5\times
10^4$ and the test error is $1.1\times 10^{-3}$. \Cref{fig:Fnus4} illustrates the NN prediction
$\lambda ^{\NN}$ and its corresponding references $\lambda $ of a sample in the test data.

\begin{figure}[htb]
  \centering
  \includegraphics[width=0.4\textwidth]{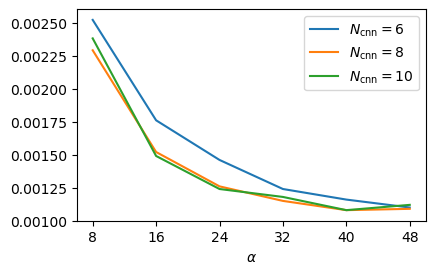}
  \label{fig:fwderr} 
  \caption{The test error for different channel numbers $\alpha$ and different
    convolution layer numbers $N_{\cnn}$.}
\end{figure}

\begin{figure}[htb]
  \centering 
  \subfloat[$\lambda (h, s)$]{
    \includegraphics[width=0.25\textwidth]{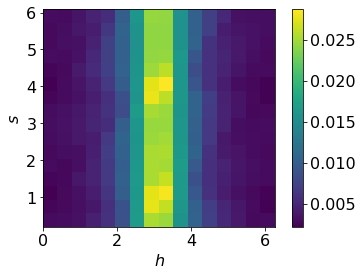}
  }
  \subfloat[$\lambda ^{\NN}(h, s)$]{
    \includegraphics[width=0.25\textwidth]{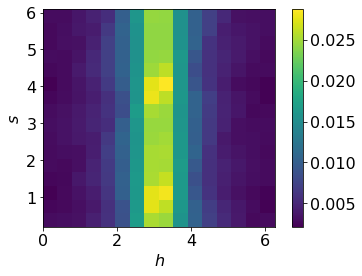}
  }
  \subfloat[$\lambda (h, s)-\lambda_{0}(h,s)$]{
    \includegraphics[width=0.25\textwidth]{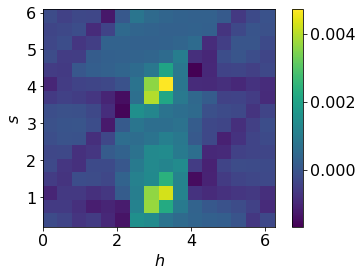}
  }
  \subfloat[$\lambda ^{\NN}(h, s)-\lambda_{0}(h,s)$]{
    \includegraphics[width=0.25\textwidth]{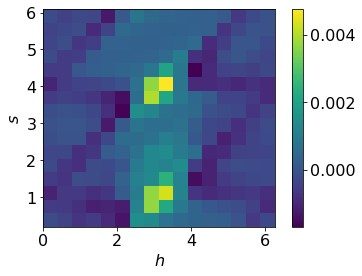}
  }
  \label{fig:Fnus4}
  \caption{Illustration of a sample in the test data for the forward problem with
    the number of ellipses $N_e=4$ in $\Omega$. The channel number $\alpha=32$ and the convolution
  layer number $N_{\cnn}=8$.}
\end{figure}

\paragraph{Inverse problem.}
Two date sets corresponding to $N_e=2,4$ are generated. The hyper-parameters in \cref{alg:inverse}
are set as $(\alpha_1=32, w_1=5, N_{\cnn_1}=6)$ and $(\alpha_2=4, w_2=3\times 3, N_{\cnn_2}=5)$ and 
the number of parameters in NN is $4.8\times 10^4$. To model the uncertainty in the measurement
data, we introduce noises to the albedo operator in the data set by defining
$\lambda_i^\delta \equiv (1+ Z_i \delta)\lambda_i$, where $Z_i$ is a Gaussian random variable with
zero mean and unity variation and $\delta$ controls the signal-to-noise ratio. For each noisy level
$\delta=0$, $0.5\%$, $1\%$, $2\%$ and $5\%$, an independent NN is trained and tested with the noisy
data set $\{(\lambda_i^\delta,\mu_i)\}$. Note that in our experiments the mean of
$\frac{\|\lambda-\lambda_{0}\|}{\lambda }$ for all the samples is about $5\%$ and hence the
signal-to-noise ratio for the difference $\lambda-\lambda_0$ is almost $100\%$ when the noise level
$\delta=5\%$.

\begin{figure}[h!]
    \centering
    \subfloat[Reference $\mu$]{
        \includegraphics[width=0.25\textwidth]{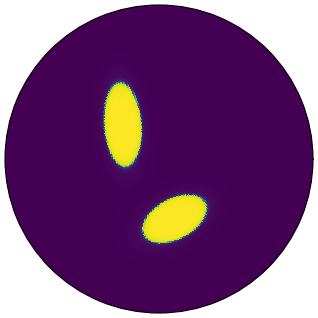}
    }\hspace{0.05\textwidth}
    \subfloat[$\mu^{\NN}$ with $\delta=0$]{
        \includegraphics[width=0.25\textwidth]{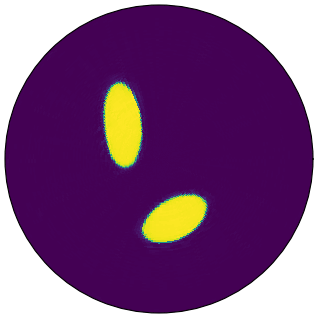}
    }\hspace{0.05\textwidth}
    \subfloat[$\mu^{\NN}$ with $\delta=0.5\%$]{
        \includegraphics[width=0.25\textwidth]{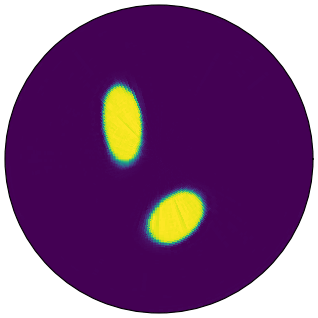}
    }\\
    \subfloat[$\mu^{\NN}$ with $\delta=1\%$]{
        \includegraphics[width=0.25\textwidth]{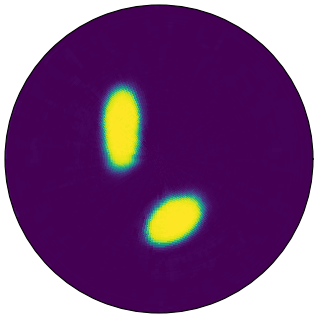}
    }\hspace{0.05\textwidth}
    \subfloat[$\mu^{\NN}$ with $\delta=2\%$]{
        \includegraphics[width=0.25\textwidth]{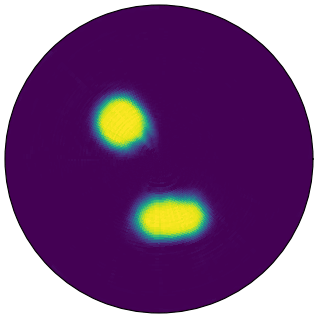}
    }\hspace{0.05\textwidth}
    \subfloat[$\mu^{\NN}$ with $\delta=5\%$]{
        \includegraphics[width=0.25\textwidth]{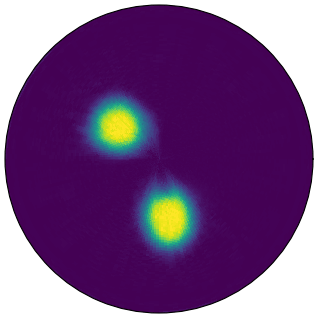}
    }
    \caption{\label{fig:nus2g9} NN prediction of a sample in the test data for the number of
      ellipses $N_e=2$ in $\Omega$ and for different noise level $\delta=0, 0.5\%, 1\%, 2\%, 5\%$.}
\end{figure}

\begin{figure}[h!]
    \centering
    \subfloat[Reference $\mu$]{
        \includegraphics[width=0.25\textwidth]{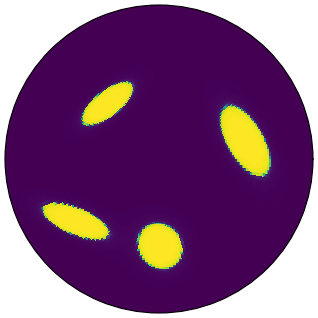}
    }\hspace{0.05\textwidth}
    \subfloat[$\mu^{\NN}$ with $\delta=0$]{
        \includegraphics[width=0.25\textwidth]{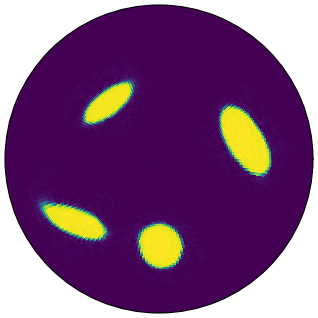}
    }\hspace{0.05\textwidth}
    \subfloat[$\mu^{\NN}$ with $\delta=0.5\%$]{
        \includegraphics[width=0.25\textwidth]{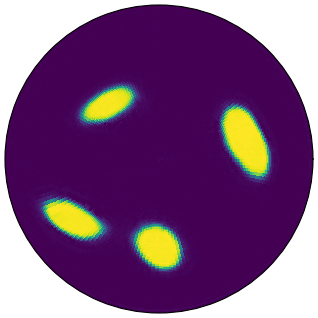}
    }\\
    \subfloat[$\mu^{\NN}$ with $\delta=1\%$]{
        \includegraphics[width=0.25\textwidth]{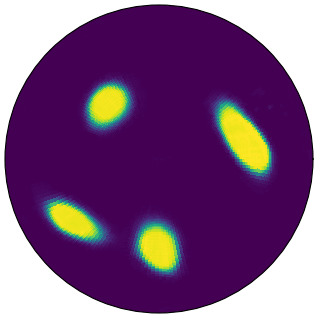}
    }\hspace{0.05\textwidth}
    \subfloat[$\mu^{\NN}$ with $\delta=2\%$]{
        \includegraphics[width=0.25\textwidth]{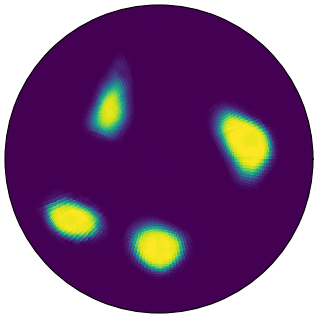}
    }\hspace{0.05\textwidth}
    \subfloat[$\mu^{\NN}$ with $\delta=5\%$]{
        \includegraphics[width=0.25\textwidth]{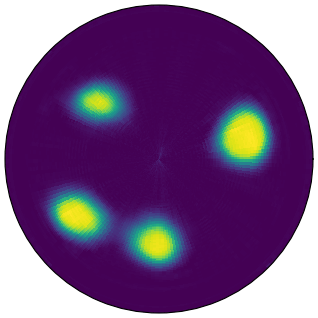}
    }
    \caption{\label{fig:nus4g9} NN prediction of a sample in the test data for the number of
      ellipses $N_e=4$ in $\Omega$ and for different noise level $\delta=0, 0.5\%, 1\%, 2\%, 5\%$.}
\end{figure}

\Cref{fig:nus2g9,fig:nus4g9} show samples in the test data for different noise level $\delta$ and
different number of ellipses $N_e$ in $\Omega$. When there is no noise in the measurement data, the
NN offers an accurate prediction of the scattering coefficient $\mu$, in the position, shape and
direction of the ellipses. For the small noise levels, for example $\delta=0.5\%$ and $1\%$, the
boundary of the shapes in the prediction is blurred while the position and direction of the ellipses
are still correct.  As the noise level $\delta$ increases, the shapes become fuzzy but the position
and number of shapes are still correctly predicted. This demonstrates the NN architecture in
\cref{alg:inverse} is capable of learning the inverse problem of OT.

To test the generalization performance of the NN, we train the NN using the data set of $N_e=2$ at a
given noise level and test the NN by the data of $N_e = 4$ with the same noise level (and vice
versa). The results, summarized in \cref{fig:generation}, indicate that the NN trained by the data,
with two inclusions is capable of recovering the measurement data of the case with four inclusions,
and vice versa. This is an indication that the trained NN is capable of predicting beyond the
training scenario.

\begin{figure}[htb]
    \centering
    \subfloat[Reference $\mu$ for $N_e=2$]{
        \includegraphics[width=0.20\textwidth]{Nus2G9Al32Ns0S5Y.png}
    }\hspace{0.02\textwidth}
    \subfloat[$\mu^{\NN}$ with $\delta=0$]{
        \includegraphics[width=0.20\textwidth]{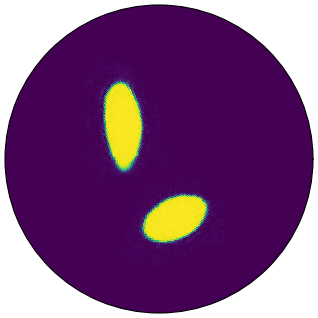}
    }\hspace{0.02\textwidth}
    \subfloat[$\mu^{\NN}$ with $\delta=0.5\%$]{
        \includegraphics[width=0.20\textwidth]{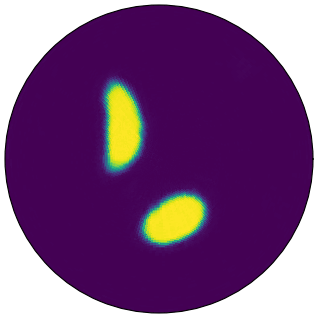}
    }\hspace{0.02\textwidth}
    \subfloat[$\mu^{\NN}$ with $\delta=1\%$]{
        \includegraphics[width=0.20\textwidth]{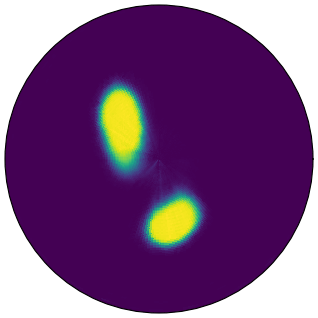}
    }\\
    \subfloat[Reference $\mu$ for $N_e=4$]{
        \includegraphics[width=0.20\textwidth]{Nus4G9Al32Ns0S7Y.png}
    }\hspace{0.02\textwidth}
    \subfloat[$\mu^{\NN}$ with $\delta=0\%$]{
        \includegraphics[width=0.20\textwidth]{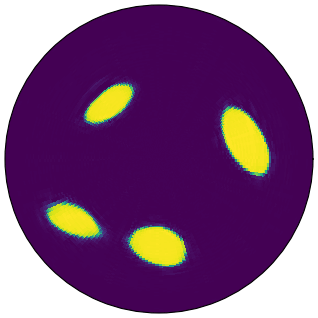}
    }\hspace{0.02\textwidth}
    \subfloat[$\mu^{\NN}$ with $\delta=0.5\%$]{
        \includegraphics[width=0.20\textwidth]{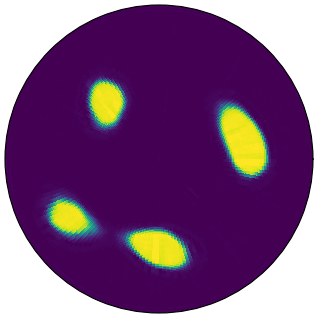}
    }
    \subfloat[$\mu^{\NN}$ with $\delta=1\%$]{
        \includegraphics[width=0.20\textwidth]{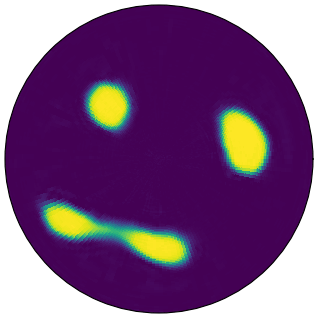}
    }
    \caption{\label{fig:generation}NN generation test. The upper figures: the NN is trained by the
      data of the number of ellipses $N_e=4$ in $\Omega$ with noise level $\delta=0$ $0.5\%$ or
      $1\%$ and test by the data of $N_e=2$ with the same noise level. The lower figures correspond
      the case of training data $N_e=2$ and test data $N_e=4$.}
\end{figure}

\section{Discussions}\label{sec:conclusion}

This paper presents a neural network approach for OT problems.  Mathematically, these NNs
approximate the forward and inverse maps between the scattering coefficient and the kernel
distribution of the albedo operator. The perturbative analysis, which indicates that the linearized
forward map can be represented by a one-dimensional convolution with multiple channels, inspires the
design of the NN architectures.

NNs have offered a few clear advantages in approximating the forward and inverse problems. For both
the forward and inverse maps, once the NN is trained, applying the map is significantly accelerated
as it only involves a single inference with the trained NN. For the inverse problem, two critical
issues for more traditional approaches are the choices of the solution algorithm and the
regularization term. NNs seem to bypass the algorithm issue by choosing an appropriate architecture
and learning the map from the data, and at the same time, identify an appropriate regularization by
automatically learning the key features from the training set.  Numerical results also demonstrate
that the proposed NNs are capable of approximating the forward and inverse maps accurately.
However, although empirically encouraging, theoretical justification of these advantages require
significant work.

%% Moreover, the NN architectures in this paper involves properties of the maps; hence it could
%% approximate the maps after well-trained. On the other hand, the NN is trained from the data, so it
%% can learn features of the data. In this sense, the NN builds a bridge between the data and PDEs,
%% which might grant the NN better approximation on the inverse problem.

The discussion in this paper focuses on the reconstruction of the scattering coefficient. Using a
similar analysis, one can extend the work to the reconstruction of the absorption coefficient or
both. The analysis in this paper can also extended to the three-dimensional OT problems by
leveraging recent work such as \cite{s.2018spherical}.

%and the technique in \cite{fan2019eit}, 

\section*{Acknowledgments}
The work of Y.F. and L.Y. is partially supported by the U.S. Department of Energy, Office of
Science, Office of Advanced Scientific Computing Research, Scientific Discovery through Advanced
Computing (SciDAC) program. The work of L.Y. is also partially supported by the National Science
Foundation under award DMS-1818449. This work is also supported AWS Cloud Credits for Research
program from Amazon.

\bibliographystyle{abbrv}
\bibliography{nn}
\end{document}